\documentclass{tlp}
\usepackage{rotating}

\usepackage{ifthen}
\usepackage{url}
\usepackage{amssymb}
\usepackage{amsmath}
\usepackage{xparse}
\usepackage{amsmath}
\usepackage[usenames,dvipsnames]{xcolor}
\usepackage{xspace}
\usepackage{datatool}
\usepackage{glossaries}

\newcommand{\m}[1]{\ensuremath{#1}\xspace}
\newcommand{\trval}[1]{\m{{\bf #1}}}

	\newcommand{\limplies}{\m{\Rightarrow}}
	
	\newcommand{\lequiv}{\m{\Leftrightarrow}}

	\newcommand{\lrule}{\m{\leftarrow}}
	\newcommand{\cause}{\m{\stackrel{c}{\lrule}}}
	
	\newcommand{\ltrue}{\trval{t}}
	\newcommand{\lfalse}{\trval{f}}
	\newcommand{\lunkn}{\trval{u}}
	\newcommand{\lincon}{\trval{i}}
	\newcommand{\bigand}{\m{\bigwedge}}
	\newcommand{\bigor}{\m{\bigvee}}

	\newcommand{\ra}{\m{\rightarrow}}

	\newcommand{\voc}{\m{\Sigma}}

	\newcommand{\struct}{\m{I}}

	\newcommand{\theory}{\m{\mathcal{T}}}

	\newcommand{\PP}{\m{\mathcal{P}}}
	\newcommand{\LL}{\m{\mathcal{L}}}

	\NewDocumentCommand\inter{g+g}{%
	  \IfNoValueTF{#1}
	    {\struct}
	    {\m{#1^{#2}}}}

	\newcommand{\XXX}{\m{\overline{X}}}

	\renewcommand{\int}{\m{\mathbb{Z}}}

	\newcommand{\leqp}{{\m{\,\leq_p\,}}}

	\DeclareMathOperator\lfp{lfp}

	\NewDocumentCommand\subs{g+g}{%
	  \IfNoValueTF{#1}
	    {\m{/}}
	    {\m{#1/ #2}}}

\newcommand{\ouracronym}[3]{%
	\newacronym{#1}{#2}{#3}
	\expandafter\newcommand\csname #1\endcsname{\gls{#1}\xspace}%
}
	\ouracronym{FO}{FO}{first-order logic}
	\ouracronym{MX}{MX}{Model Expansion}
	\ouracronym{MO}{MO}{Model Optimization}
	\ouracronym{ASP}{ASP}{Answer Set Programming}
	\ouracronym{TP}{TP}{Theorem Proving}
	\ouracronym{LP}{LP}{Logic Programming}
	\ouracronym{CP}{CP}{Constraint Programming}
	\ouracronym{FP}{FP}{Functional Programming}
	\ouracronym{KR}{KR}{Knowledge Representation}
	\ouracronym{CSP}{CSP}{Constraint Satisfaction Problem}
	\ouracronym{SMT}{SMT}{SAT Modulo Theories}
	\ouracronym{KBS}{KBS}{Knowledge Base System}
	\ouracronym{NNF}{NNF}{Negation Normal Form}
	\ouracronym{FNNF}{FNNF}{Flat Negation Normal Form}
	\ouracronym{DefNNF}{DefNNF}{Definition Negation Normal Form}
	\ouracronym{CDCL}{CDCL}{Conflict-Driven Clause-Learning}
	\ouracronym{LCG}{LCG}{Lazy Clause Generation}

	\makeatletter
	\def\ifenv#1{
	\def\@tempa{#1}%
	\def\@ttempa{#1*}%
	\ifx\@tempa\@currenvir
	\expandafter\@firstoftwo
	\else
	\expandafter\@secondoftwo
	\fi
	}
	\makeatother

	\newcommand{\ddrule}[4]{\ensuremath{#1 \leftarrow #2 & \{#3\} & #4}}
	\newcommand{\drule}[2]{\ensuremath{#1 & \leftarrow & #2}}

	\newcommand{\darule}[4]{\ensuremath{#1 \leftarrow #2 & \{#3\} & #4}}
	\newcommand{\arule}[2]{\ensuremath{#1 \, &\leftarrow \, #2}}

	\newenvironment{ldef}{\renewcommand\arraystretch{1.25}\left\{\begin{array}{l@{ \,}l@{\,}l}}{\end{array}\right\}}

	\newcommand{\LNDRule}[2]{
	\ifenv{array}
	{\drule{#1}{#2}}
	{ \ifenv{align}
		{\arule{#1}{#2}}
		{\ifenv{align*}
		{\arule{#1}{#2}}
		{ERROR: using LDRule in unsupported environment: \@currenvir}
		}
	}
	}

	\newcommand{\LDRule}[4]{
	\ifenv{array}
	{\ddrule{#1}{#2}{#3}{#4}}
	{ \ifenv{align}
		{\darule{#1}{#2}{#3}{#4}}
		{\ifenv{align*}
		{\darule{#1}{#2}{#3}{#4}}
		{ERROR: using LDRule in unsupported environment: \@currenvir}
		}
	}
	}

	\NewDocumentCommand\LRule{m+g+g+g}{%
		\IfNoValueTF{#2}%
		{#1.&}{%
		\IfNoValueTF{#3}
		{\LNDRule{#1}{#2.}}
		{\LDRule{#1}{#2.}{#3}{#4}}%
		}
	}

	\NewDocumentCommand\CLRule{m+g}{%
	\ifenv{array}
	{\cdrule{#1}{#2}}
	{ \ifenv{align}
		{\carule{#1}{#2}}
		{\ifenv{align*}
			{\carule{#1}{#2}}
			{ERROR: using CLRule in unsupported environment: \@currenvir}
		}
	}
	}

	\NewDocumentCommand\carule{m+g}{%
		\IfNoValueTF{#2}
			{\ensuremath{#1.}}
			{\ensuremath{#1 \, &\cause \, #2}}}
	\NewDocumentCommand\cdrule{m+g}{%
		\IfNoValueTF{#2}
			{\ensuremath{#1.}}
			{\ensuremath{#1 & \cause & #2}}}

	\newcommand{\algrule}[4]{
	\hbox{{#1}:}& 
	\quad #2 ~\longrightarrow~ #3 
	\hbox{~ if } #4\\
	}

	\newcommand{\AlgoRule}[4]{
	\ifenv{array}
	{\algrule{#1}{#2}{#3}{#4}}
		{ERROR: using AlgoRule in unsupported environment: \@currenvir}
	}

\newcommand{\commentstyle}{\color{Gray}}

	\usepackage[final]{listings}

	\lstdefinelanguage{idp}{
		morekeywords=[1]{namespace,vocabulary,theory,structure,procedure,term,set,formula, spec, specification,query},
		morekeywords=[2]{include,using,type,isa,contains,partial,extern,LFD,GFD,constructed,from,constraint,pred,supertype,of,subtype,define},
		morekeywords=[3]{int,float,char,string,nat},
		morekeywords=[4]{if,then,else,for,end},
		morecomment=[s]{/*}{*/},	
		morecomment=[l]{//}
	}
	\newcommand{\ignore}[1]{}

	\newboolean{nocomments}
	\setboolean{nocomments}{false}

	\newboolean{commentmargin}
	\setboolean{commentmargin}{true}

	\newcommand{\namedcomment}[3]{
		\ifthenelse{\boolean{nocomments}}
		{} 
		{ 
			\ifthenelse{\boolean{commentmargin}}
				{ {\color{#3} \marginpar{\color{#3}\sc #2}#1}  } 
				{  {\color{#3} {\sc #2}: #1}  } 
		}
	}
	\newcommand{\mnamedcomment}[3]{\ifthenelse{\boolean{nocomments}}{}{{\marginpar{ \color{#3}{\sc #2}:#1}}}}

	\usepackage{soul}

\usepackage{wrapfig}
\newtheorem{theorem}{Theorem}[section]
\newtheorem{definition}[theorem]{Definition} 
\newtheorem{proposition}[theorem]{Proposition} 
\newtheorem{example}[theorem]{Example} 
\newtheorem{lemma}[theorem]{Lemma} 
\usepackage{natbib}
\renewcommand\cite[1]{\citep{#1}}

\usepackage[all,cmtip]{xy}
\usepackage{pgffor}

\usepackage{color}
\usepackage{graphicx}

\usepackage{etoolbox}

\newcommand\setcitation[2]{%
  \csdef{mycommoncitation#1}{#2}}
\newcommand\getcitation[1]{%
  \csuse{mycommoncitation#1}}

\setcitation{IDP}{WarrenBook/DeCatBBD14}
\setcitation{fodot}{tocl/DeneckerT08}
\setcitation{CP}{Apt03}
\setcitation{EZCSP}{lpnmr/Balduccini11}
\setcitation{KR}{Baral:2003}
\setcitation{ASPComp2}{lpnmr/DeneckerVBGT09}
\setcitation{ASPComp3}{journals/tplp/CalimeriIR14}
\setcitation{ASPComp4}{conf/lpnmr/AlvianoCCDDIKKOPPRRSSSWX13}
\setcitation{CPSupport}{ictai/DeCat13}
\setcitation{functionDetection}{iclp/DeCatB13}
\setcitation{fodot2asp}{DeneckerLTV12} 
\setcitation{Tarskian}{DeneckerLTV12} 
\setcitation{Inca}{iclp/DrescherW12}
\setcitation{csp2asp}{ijcai/DrescherW11}
\setcitation{DPLLT}{cav/GanzingerHNOT04}
\setcitation{AspInPractice}{synthesis/2012Gebser}
\setcitation{clasp}{ai/GebserKS12}
\setcitation{oclingo}{kr/GebserGKOSS12}
\setcitation{gringo}{lpnmr/GebserST07}
\setcitation{cmodels}{aaai/GiunchigliaLM04}
\setcitation{inputster}{tplp/Jansen13}
\setcitation{DLV}{tocl/LeonePFEGPS06}
\setcitation{LearningPaper}{TPLP/BruynoogheBBDDJLRDV} 
\setcitation{clog}{iclp/BogaertsVDV14} 
\setcitation{foc}{iclp/BogaertsVDV14}
\setcitation{inferenceClog}{ecai/BogaertsVDV14}
\setcitation{examplesClog}{nmr/BogaertsVDV14b} 
\setcitation{AFT}{DeneckerMT00}
\setcitation{KBS}{iclp/DeneckerV08}
\setcitation{KBPE}{inap/DePooterWD11}
\setcitation{lazyGrounding}{corr/CatDSB14} 
\setcitation{ASP}{marek99stable}
\setcitation{satid}{sat/MarienWDB08}
\setcitation{lazyclausegeneration}{constraints/OhrimenkoSC09}
\setcitation{FP}{ACMCS/Hudak89}
\setcitation{GroundingWithBounds}{jair/WittocxMD10}
\setcitation{SAT}{faia/SilvaLM09}
\setcitation{LTC}{iclp/Bogaerts14}
\setcitation{SPSAT}{ictai/DevriendtBMDD12}
\setcitation{BreakID}{cspsat/DevriendtBB14}
\setcitation{LCG}{stuckeyLCG}
\setcitation{MiniZinc}{conf/cp/NethercoteSBBDT07}
\setcitation{amadini}{cpaior/AmadiniGM13}
\setcitation{bootstrapping}{lash/BogaertsJDJBD14}
\setcitation{GroundedFixpoints}{ai/BogaertsVD15}
\setcitation{LogicBlox}{datalog/GreenAK12}
\setcitation{proB}{journals/sttt/LeuschelB08}
\setcitation{NaturalInductions}{KR/DeneckerV14} 
\setcitation{LP}{jacm/EmdenK76}
\setcitation{SMT}{faia/BarrettSST09}
\setcitation{AF}{ai/Dung95}
\setcitation{ADF}{kr/BrewkaW10}
\setcitation{ADFRevisited}{ijcai/BrewkaSEWW13}
\setcitation{DefaultLogic}{ai/Reiter80}
\setcitation{DL}{ai/Reiter80}
\setcitation{AEL}{mo85}
\setcitation{}{}
\setcitation{completion}{adbt/Clark78}
\setcitation{}{}
\setcitation{}{}
\setcitation{}{}
\setcitation{}{}
\setcitation{}{}
\setcitation{}{}
\setcitation{}{}
\setcitation{}{}
\setcitation{}{}
\setcitation{}{}
\setcitation{}{}
\setcitation{}{}
\setcitation{}{}
\setcitation{}{}
\setcitation{}{}
\setcitation{}{}
\setcitation{}{}
\setcitation{}{}
\setcitation{}{}
\setcitation{}{}
\setcitation{}{}
\setcitation{}{}
\setcitation{}{}
\setcitation{}{}
\setcitation{}{}
\setcitation{}{}
\setcitation{}{}
\setcitation{}{}
\setcitation{}{}
\setcitation{}{}
\setcitation{}{}
\setcitation{}{}
\setcitation{}{}
\setcitation{}{}
\setcitation{}{}
\setcitation{}{}
\setcitation{}{}
\setcitation{}{}
\setcitation{}{}
\setcitation{}{}
\setcitation{}{}
\setcitation{}{}
\setcitation{}{}
\setcitation{}{}
\setcitation{}{}
\setcitation{}{}
\setcitation{}{}
\setcitation{}{}
  
\newcommand\refto[1]{%
      \getcitation{#1}}
      
\newcommand\mycite[1]{%
      \ifcsname mycommoncitation#1\endcsname%
   \cite{\getcitation{#1}}%
  \else%
    \cite{#1}
  \fi%
}

\begin{document}
%
\title[Knowledge Compilation of Logic Programs Using AFT]{
Knowledge Compilation of Logic Programs Using Approximation Fixpoint Theory
}

\author[B.~Bogaerts and G.~Van den Broeck]{Bart Bogaerts and Guy Van den Broeck\\
Department of Computer Science, KU Leuven, Belgium\\
\email{bart.bogaerts@cs.kuleuven.be,guy.vandenbroeck@cs.kuleuven.be}}

 
\maketitle

\newcommand\pstruct{\m{\mathcal{I}}}
\newcommand\compile{\m{\textsc{Compile}}}
\newcommand\WMC{\m{\textsc{WMC}}}

	\setboolean{commentmargin}{false}
	
\newboolean{showproofs}
\setboolean{showproofs}{false}
\newcommand\erasableproof[1]{\ifthenelse{\boolean{showproofs}}{#1}{}}

\newcounter{numberOfStoredProofs}
\stepcounter{numberOfStoredProofs}

\newcommand\proofintext[1]{#1}
\renewcommand\proofintext[1]{}

\newboolean{displayproof}
\newcommand\thmwithproof[5]{
\thmwithproofgeneral{#1}{#2}{#3}{#4}{#5}{false}
}
\newcommand\thmwithproofhere[5]{
\thmwithproofgeneral{#1}{#2}{#3}{#4}{#5}{true}
}
\newcommand\thmwithproofgeneral[6]{
\begin{#3}\label{#1}#4
\end{#3}
\setboolean{displayproof}{#6}
\ifthenelse{\boolean{displayproof}}
{\begin{proof}#5\end{proof}} 
{ \proofintext{ \begin{proof} 	#5 \end{proof}}
\expandafter\newcommand\csname mystoredproof\the\value{numberOfStoredProofs} \endcsname{\ \newline\noindent{\it #2 \ref{#1}.} {\newline #4}\ \newline \begin{proof}#5\end{proof}}
\stepcounter{numberOfStoredProofs}
}
}

\newcommand\inproofappendix[1]{
\expandafter\newcommand\csname mystoredproof\the\value{numberOfStoredProofs} \endcsname{#1}
\stepcounter{numberOfStoredProofs}
}

\newcommand\getproof[1]{
\ifcsname mystoredproof#1 \endcsname%
	\csname mystoredproof#1 \endcsname%
  \else%
  \fi%
}

\newcommand{\proofs}{
\foreach \n in {0,...,\value{numberOfStoredProofs}}{\getproof{\n}}
}

\newcommand{\plat}{\m{L_{p}}}
\newcommand{\dplat}{\m{L^{d}_{p}}}
\newcommand{\sstruct}{\m{\mathcal{A}}}
\newcommand{\spstruct}{\m{\mathcal{S}}}
\newcommand{\proj}{\m{\pi}}
\newcommand{\imcons}{\m{T_\PP}}
\newcommand{\pimcons}{\m{\mathcal{T}_\PP}}
\newcommand{\ppimcons}{\m{\varPsi_\PP}}

\newcommand\tcg{\m{\mathit{r}}}

\newcommand{\guy}[1]{\namedcomment{#1}{GVdB}{OliveGreen}}

\newcommand\Th{\m{\mathit{Th}}}

\newcommand{\modelswfm}{\m{\models_{\mathit{wf}}}}

\begin{abstract}
Recent advances in knowledge compilation introduced techniques to compile \emph{positive} logic programs into propositional logic, essentially exploiting the constructive nature of the least fixpoint computation.
This approach has several advantages over existing approaches: it maintains logical equivalence, does not require (expensive) loop-breaking preprocessing or the introduction of auxiliary variables, and significantly outperforms existing algorithms.
Unfortunately, this technique is limited to \emph{negation-free} programs. 
In this paper, we show how to extend it to general logic programs under the well-founded semantics. 

We develop our work in approximation fixpoint theory, an algebraical framework that unifies semantics of different logics. 
As such, our algebraical results are also applicable to autoepistemic logic, default logic and abstract dialectical frameworks.

\end{abstract}


\section{Introduction}\label{sec:intro}




There is a fundamental tension between the expressive power of a knowledge representation language, and its support for efficient reasoning. Knowledge compilation studies this tension~\cite{cadoli1997survey,darwiche2002knowledge}, by identifying \emph{languages} that support certain queries and transformations efficiently. It studies the relative succinctness of these languages, and is concerned with building \emph{compilers} that can transform knowledge bases into a desired target language.
For example, after compiling two CNF sentences into the OBDD language~\cite{Bryant86}, their equivalence can be checked in polynomial time.
Applications of knowledge compilation are found in diagnosis~\cite{huang2005compiling},
databases~\cite{suciu2011probabilistic}, planning~\cite{palacios2005pruning}, graphical models~\cite{chavira2005compiling,tplp/FierensBRSGTJR15} and machine learning~\cite{lowd2008learning}.
These techniques are most effective when the cost of compilation can be amortised over many queries to the knowledge base. 

Knowledge compilation has traditionally focused on subsets of propositional logic and Boolean circuits in particular~\cite{darwiche2002knowledge,ijcai/Darwiche11}.
Logic programs have received much less attention, which is surprising given their historical significance in AI and current popularity in the form of answer set programming (ASP)~\mycite{ASP}.
Closest in spirit are techniques to encode logic programs into CNF~\cite{DBLP:journals/amai/Ben-EliyahuD94,ijcai/LinZ03,ai/LinZ04,ecai/Janhunen04,jancl/Janhunen06}. A notable difference with traditional knowledge compilation is that many of these encodings are task-specific: the resulting CNF is not equivalent to the logic program. Instead, it is \emph{equisatisfiable} for the purpose of satisfiability checking, or has an identical model count for the purpose of probabilistic inference~\cite{tplp/FierensBRSGTJR15}.\footnote{Probabilistic inference on the CNF may itself perform a second knowledge compilation step.} 
These encodings often introduce new variables and loop-breaking formulas, which blow up the representation. \citet{tocl/LifschitzR06} showed that there can be no polynomial translation of ASP into a flat propositional logic theory without auxiliary variables.\footnote{Similar, task-specific, translation techniques of logic programs into difference logic \cite{lpnmr/JanhunenNS09} and ordered completion \cite{ai/AsuncionLZZ12} exist. 	}

Recently, \citet{Jonas} introduced a novel knowledge compilation technique for \emph{positive} logic programs. 
As an example, consider the logic program $\PP$ defining the transitive closure $\tcg$ of a binary relation $e$:
\[
 \begin{ldef}
  \forall X,Y: \tcg(X,Y) &\lrule e(X,Y).\\
  \forall X,Y,Z: \tcg(X,Y) &\lrule e(X,Z) \land \tcg(Z,Y).
 \end{ldef}
\]
Intuitively, \citet{Jonas} compute the minimal model of $\PP$ for all interpretations of $e(\cdot,\cdot)$ \emph{simultaneously}. 
They define a lifted least fixpoint computation where the intermediate results are symbolic interpretations of $\tcg(\cdot,\cdot)$ in terms of $e(\cdot,\cdot)$. For example, in a domain $\{a,b,c\}$, the interpretation of $\tcg(a,b)$ in the different steps of the least fixpoint computation would be. 
\begin{align*}
\tcg(a,b)&: &	\lfalse &\qquad\rightsquigarrow &e(a,b)&\qquad\rightsquigarrow & e(a,b) \lor (e(a,c)\land e(c,b)) 					
\end{align*}
I.e., initially, $r(a,b)$ is false; next $r(a,b)$ is derived to be true if $e(a,b)$ holds; finally, $r(a,b)$ also holds if $e(a,c)$ and $e(c,b)$ hold. 
The result of this sequence is a symbolic, Boolean formula representation of the well-founded model for each interpretation of $e$; this formula can be used for various inference tasks.
This approach has several advantages over traditional knowledge compilation methods: it preserves logical equivalence\footnote{In the sense that an interpretation is a model of the resulting propositional theory if and only if it is a model of the given logic program under the parametrised well-founded semantics.} (and hence, enables us to port any form of inference---e.g., abductive or inductive reasoning, (weighted) model counting, query answering, \dots) and does not require (expensive) loop-breaking preprocessing or auxiliary variables. 
\citet{Jonas} showed that this method for compiling positive programs (into the SDD language~\cite{ijcai/Darwiche11}) significantly outperforms traditional approaches that compile the completion of the program with added loop-breaking formulas.

Unfortunately, the methods of \citet{Jonas} do not work in the presence of negation, i.e., if the immediate consequence operator is non-monotone. 
In this paper, we show how the well-founded model computation from \citet{GelderRS91}, that works on partial interpretations, can be executed symbolically, resulting in the \emph{parametrised well-founded model}. 
By doing this, we essentially compute the well-founded model of an \emph{exponential} number of logic programs at once. 
%

Our algorithm works in principle on any representation of Boolean formulas; we study complexity for this algorithm taking Boolean circuits as target language; in this case we find that our algorithm has \emph{polynomial} time complexity.
General Boolean circuits are not considered to be an interesting target language, as they are not tractable for any query of interest. 
However, what we achieve here is a \emph{change of semantic paradigm} that uncovers all the machinery for propositional logic (SAT solvers, model counters, etc.). It is a required step before further compiling the circuit into a language such as OBDD or SDD, which do permit tractable querying. 
It is also possible to encode the circuit into CNF, similar to \citet{ecai/Janhunen04}.
There is a long list of queries and transformations that become supported on logic programs (under the well-founded semantics), by virtue of our algorithm. After a transformation to propositional logic, we can use standard tools to check whether one logic program is entailed by another, find models that are minimal with respect to some optimisation term, check satisfiability, count or enumerate models, and forget or condition variables~\cite{darwiche2002knowledge}.
For example, the following definition of the transitive closure of $e$ syntactically differs from the previous. 
\[
 \begin{ldef}
  \forall X,Y: \tcg(X,Y) &\lrule e(X,Y).\\
  \forall X,Y,Z: \tcg(X,Y) &\lrule \tcg(X,Z) \land \tcg(Z,Y).
 \end{ldef}
\]
With our algorithm, we can compile both programs into an OBDD representation. On these OBDDs, we can verify the equivalence of the logic programs using existing OBDD algorithms. As logic programs under the well-founded semantics encode \emph{inductive definitions} \mycite{NaturalInductions}, we now have the machinery to check that two definitions define the same concept for each interpretation of the parameters ($e$ in our example).
Moreover, our algorithm can be stopped at any time to obtain upper and lower bounds on the fixpoint, which gives us \emph{approximate knowledge compilation} for logic programs~\cite{selman1996knowledge}.

The original motivation for this research is the fact that probabilistic inference tools such as ProbLog~\cite{tplp/FierensBRSGTJR15} use knowledge compilation for probabilistic inference by (weighted) model counting; they compile a logic program into a d-DNNF or SDD (with auxiliary variables) and subsequently calling a weighted model counter. Vlasselaer et al.\ showed that for \emph{positive} logic programs, this can be done much more efficiently using bottom-up compilation techniques. 
We extend these techniques to general logic programs to capture the full ProbLog language. 


More generally, we develop our ideas in \emph{approximation fixpoint theory} (AFT), an abstract algebraical theory that captures all common semantics of logic programming, autoepistemic logic, default logic, Dung's argumentation frameworks and abstract dialectical frameworks (as shown by \citet{\refto{AFT}} and \citet{journals/ai/Strass13}). 
Afterwards, we show how the algebraical results apply to logic programming. We thus extend the ideas by \citet{Jonas} in two ways; first, by developing a theory that works for \emph{general} logic programs and secondly by lifting the theory to the algebraical level. 
Due to the high level of abstraction, our proofs are (relatively) compact and our algebraical results are immediately applicable to all aforementioned paradigms.
Due to page restrictions, proofs are postponed to the online appendix (Appendix B) and we only apply our theory to logic programming.

%
%

Summarised, the main contributions of this paper are as follows:
\textit{(i)} we present the algebraical foundations for a novel knowledge compilation technique for \emph{general} logic programs, 
\textit{(ii)} we apply the algebraical theory to logic programming, resulting in a family of equivalence-preserving algorithms,
\textit{(iii)} we show that Boolean circuits are at least as succinct as propositional logic programs (under the parametrised well-founded semantics), and 
\textit{(iv)} we pave the way towards knowledge compilation for other non-monotonic formalisms, such as autoepistemic logic. 

%

%
\ignore{
As a running example, consider the logic program $\PP$ defining the transitive closure $\tcg$ of a binary relation $E$:
\[
 \begin{ldef}
  \forall x, y: \tcg(x,y) &&\lrule e(x,y).\\
  \forall x, y: \tcg(x,y) \lrule \exists z: e(x,z) \land \tcg(z,y).
 \end{ldef}
\]
Encoding $\PP$ into CNF naively (i.e., without any optimisations) for a given set of nodes will result in the completion of $\PP$ together with sentences of the form\footnote{For readability, the loop-breaking formulas are not in CNF.}
\begin{align*}
 \tcg(n_1,n_2) &\limplies  e(n_1,n_2) && \lor (e(n_1,n_3)  \land \tcg(n_3,n_2) \land D_{n_1,n_2,n_3,n_2})\\ &&& \lor (e(n_1,n_4)  \land \tcg(n_4,n_2) \land D_{n_1,n_2,n_4,n_2}) \lor \dots
\end{align*}
The  $D$-atoms are newly introduced symbols with the intended interpretation that $D_{a,b,c,d}$ holds if the fact that $(a,b)$ is in the transitive closure \emph{depends} on the fact that $(c,d)$ is in the transitive closure, i.e., that $\tcg(a,b)$ has a greater level than $\tcg(c,d)$ in the level numbering \cite{ecai/Janhunen04}.\footnote{\citet{ecai/Janhunen04} uses a slightly more complex transformation in order to obtain a translation satisfying certain desirable criteria, such as an identical model count.} 
The translated CNF also enforces that $D$ is asymmetric and transitive.

In this paper, we introduce a framework to efficiently compile \emph{parametrised logic programs} under the well-founded semantics into equivalent \emph{Boolean circuits}.
Our approach is based on the following observations. 
For every interpretation $I$ of the \emph{parameters}, or \emph{extensional predicates}, of \PP (in this case, for $E$), the above logic program induces an immediate consequence operator $T_{\PP}^I$. 
Since this is a monotone operator (it is a positive logic program), the well-founded model is the least fixpoint of this operator, e.g., if $E=\{(a,b),(b,c)\}$, it is the limit of the sequence 
\[\emptyset\to\{\tcg(a,b),\tcg(b,c)\}\to\{\tcg(a,b),\tcg(b,c),\tcg(a,c)\}.\]
For every interpretation $I$ of $E$, we can compute a similar sequence. 
Instead of doing similar computations all over again (an exponential number of times), we define a new operator \pimcons that ``summarises'' $T_{\PP}^I$ for all $I$, and that works on abstract interpretations in which every atom is interpreted by a propositional formula over the parameters of \PP. 
We construct the extended operator $\pimcons$ in such a way that for every $I$, the least fixpoint of $T_{\PP}^I$ can be obtained from the least fixpoint of $\pimcons$.
}
\section{Preliminaries}\label{sec:prelims}


\subsection{Lattices and Approximation Fixpoint Theory}

A \emph{complete lattice} $\langle L,\leq\rangle$ is a set $L$ equipped with a partial order $\leq$ such that every subset $S$ of $L$ has a \emph{least upper bound},  denoted $\bigor S$ and a \emph{greatest lower bound}, denoted $\bigand S$. 
If $x$ and $y$ are two lattice elements, we use the notations 
 $x\land y=\bigand \{x,y\}$ and $x\lor y=\bigor \{x,y\}$.
A complete lattice has a least element $\bot$ and a greatest element $\top$. 
An operator $O:L\to L$ is \emph{monotone} if $x\leq y$ implies that $O(x)\leq O(y)$. 
Every monotone operator $O$ in a 
complete lattice has a least fixpoint, denoted $\lfp(O)$.
A mapping $f:(L,\leq_L)$ $\to(K,\leq_K)$ between lattices is a \emph{lattice morphism} if it preserves least upper bounds and greatest lower bounds, i.e. if for every subset $X$ of $L$, $f(\bigor X)= \bigor f(X)$ and  $f(\bigand X)= \bigand f(X)$.

Given a lattice, approximation fixpoint theory makes uses of the bilattice 
$L^2$.  We define \emph{projections} as usual:
$(x,y)_1=x$ and $(x,y)_2=y$.  Pairs $(x,y)\in L^2$ are used to
approximate all elements in the interval $[x,y] = \{z\mid x\leq
z\wedge z\leq y\}$. We call $(x,y)\in L^2$ \emph{consistent} if $x\leq
y$, that is, if $[x,y]$ is non-empty. We use $L^c$ to denote the set
of consistent pairs. Pairs $(x,x) $ are called
\emph{exact}.  
The \emph{precision
  ordering} on $L^2$ is defined as $(x,y) \leqp (u,v)$ if $x\leq u$
and $v\leq y$. In case $(u,v)$ is consistent,  $(x,y)$ is less precise than $(u,v)$ if $(x,y)$
approximates all elements approximated by $(u,v)$, or in other words
if $[u,v]\subseteq [x,y]$.  If $L$ is a complete lattice, then so is
$\langle L^2,\leqp\rangle$. 

AFT studies fixpoints of operators $O:L\ra L$ through operators approximating $O$.
 An operator $A: L^2\to L^2$  is an \emph{approximator} of $O$ if it is \leqp-monotone,  and has the property that for all $x$, $O(x)\in A(x,x)$. 
Approximators are
internal in $L^c$ (i.e., map $L^c$ into $L^c$).
As usual, we restrict our attention to \emph{symmetric} approximators: approximators $A$ such that for all $x$ and $y$, $A(x,y)_1 = A(y,x)_2$.
\citet{DeneckerMT04} showed that the consistent fixpoints of interest are uniquely determined by an approximator's restriction to $L^c$, hence, we only define approximators on $L^c$. 


AFT studies fixpoints of $O$ using fixpoints of $A$. 
The $A$-Kripke-Kleene fixpoint is the $\leqp$-least fixpoint of $A$ and has the property that it approximates all fixpoints of $O$. 
A partial $A$-stable fixpoint is a pair  $(x,y)$ such that $x=\lfp(A(\cdot,y)_1)$ and $y=\lfp(A(x,\cdot)_2)$. The $A$-well-founded fixpoint is the least precise partial $A$-stable fixpoint. 
 An \emph{$A$-stable fixpoint} of $O$ is a fixpoint $x$ of $O$ such that $(x,x)$ is a partial $A$-stable fixpoint. 
%
%
The $A$-Kripke-Kleene fixpoint of $O$ can be constructed by iteratively applying $A$, starting from $(\bot,\top)$. 
For the $A$-well-founded fixpoint, \citet{lpnmr/DeneckerV07} worked out a similar constructive characterisation as follows.
 
An \emph{$A$-refinement} of $(x,y)$ is a pair $(x',y')\in L^2$ satisfying one of the following conditions \textit{(i)} $(x,y)\leqp(x',y')\leqp A(x,y)$, or \textit{(ii)}
$x'=x$ and  $A(x,y')_2\leq y'\leq y$. 
An $A$-refinement is \emph{strict} if $(x,y)\neq (x',y')$.
%
We call refinements of the first kind \emph{application refinements} and refinements of the second kind \emph{unfoundedness refinements}. 
%
%
 A \emph{well-founded induction} of $A$  is a sequence 
$(x_i,y_i)_{i\leq \beta}$
with $\beta$ an ordinal such that 
\begin{itemize}
	\item $(x_0,y_0) = (\bot,\top)$;
	\item $(x_{i+1},y_{i+1})$ is an A-refinement of $(x_{i},y_{i})$, for  all $i<\beta$;
	\item $(x_{\lambda},y_{\lambda})= \bigor_\leqp \{(x_i,y_i)\mid i<\lambda\}$
	      for each limit ordinal $\lambda\leq\beta$.
\end{itemize}
A well-founded induction is \emph{terminal} if its limit $(x_\beta,y_\beta)$ has no strict $A$-refinements.
%
For a given approximator $A$, there are many different terminal well-founded inductions of $A$.
\citet{lpnmr/DeneckerV07} showed that they all have the same limit, which equals the $A$-well-founded fixpoint of $O$. 
\citet{lpnmr/DeneckerV07} also showed how to obtain maximally precise unfoundedness refinements. 
\begin{proposition}[\citeauthor{lpnmr/DeneckerV07},~\citeyear{lpnmr/DeneckerV07}]\label{prop:biggestufs}
Let $A$ be an approximator of $O$ and $(x,y)\in L^2$. 
Let $S_A^x$ be the operator on $L$ that maps every $y'$ to $A(x,y')_2$.
This operator is monotone. 
The smallest $y'$ such that $(x,y')$ is an unfoundedness refinement of $(x,y)$ is given by 
$y'=\lfp (S_A^x)$.
 \end{proposition}


\subsection{Logic Programming} \label{ss:lp}

In this paper, we restrict our attention to propositional logic programs.
However, AFT has been applied in a much broader context \cite{DeneckerMT00,tplp/PelovDB07,lpnmr/AnticEF13} and our results apply in these richer settings as well. 

Let \voc be an alphabet, i.e., a collection of symbols  called \emph{atoms}.
 A \emph{literal} is an atom $p$ or its negation $\lnot p$.  
A logic program $\mathcal{P}$ is a set of \emph{rules} $r$ of the form 
$h\lrule l_1\land l_2\land \dots \land l_n$, where
$h$ is an atom called the \emph{head} of $r$, denoted $\mathit{head}(r)$, and the $l_i$ are literals. The formula $l_1\land l_2\land \dots \land l_n$ is the \emph{body} of $r$, denoted $\mathit{body}(r)$. A rule $r=\forall \XXX: h\lrule\varphi$ is, as usual, a shorthand for the \emph{grounding of $r$}, the collection of rules obtained by substituting the variables \XXX by elements from a given domain. 
If $p\in \voc$, the formula $\varphi_p$ is $\bigvee_{r\in \PP\land \mathit{head}(r)=p} \mathit{body}(r)$.
An interpretation $\struct$ of the alphabet \voc is an element of $2^\voc$, i.e., a subset of $\voc$.  The set of interpretations $2^\voc$ forms a lattice equipped with the order $\subseteq$. 
The truth value (\ltrue or \lfalse) of a propositional formula $\varphi$ in a structure $\struct$, denoted $\varphi^\struct$ is defined as usual.
With a logic program \PP, we associate an immediate consequence operator \cite{jacm/EmdenK76} $T_\PP$ mapping structure $\struct$ to 
$T_\PP(\struct) = \{p\mid \varphi_p^\struct=\ltrue\}$.

In the context of logic programming, elements of the bilattice $\left(2^\voc\right)^2$ are four-valued interpretations, pairs $\pstruct= (\struct_1,\struct_2)$ of interpretations.
A four-valued interpretation maps atoms $p\in \voc$ to tuples of two truth values $(p^{\struct_1}, p^{\struct_2})$. Such tuples are often identified with four-valued truth values (true ($\ltrue$), false ($\lfalse$), unknown ($\lunkn$) and inconsistent ($\lincon$)). Intuitively, $p^{\struct_1}$ represents whether $p$ is true, and $p^{\struct_2}$ whether $p$ is possible, i.e., not false. Thus, the following correspondence holds $\ltrue=(\ltrue,\ltrue), \lfalse=(\lfalse,\lfalse), \lunkn=(\lfalse,\ltrue)$ (and $\lincon=(\ltrue,\lfalse)$).
The pair $(\struct_1,\struct_2)$ approximates all interpretations $\struct'$ with $\struct_1\subseteq \struct'\subseteq \struct_2$.
We are mostly concerned with consistent (also called partial) interpretations: tuples $(\struct_1,\struct_2)$ with $\struct_1\subseteq\struct_2$, i.e., interpretations that map no atoms to $\lincon$. 
If $\pstruct$ is a partial interpretation, and $\varphi$ a formula, we write $\varphi^\pstruct$ for the standard three-valued valuation based on Kleene's truth tables \cite{Kleene38}.
We often identify interpretation $I$ with the partial interpretation $(I,I)$.

The most common approximator for logic programs is Fitting's (\citeyear{tcs/Fitting02}) immediate consequence operator $\Psi_\PP$ , a 
generalisation of $T_\PP$ to partial interpretations: 
\begin{align*}
    \Psi_\PP(\pstruct)_1 &=\{a\in \voc\mid \exists r\in \PP: body(r)^\pstruct=\ltrue \land head(r)=a\},\\
    \Psi_\PP(\pstruct)_2 &=\{a\in \voc\mid\exists r\in \PP:  body(r)^\pstruct\neq \lfalse \land head(r)=a\}
   \end{align*}
\citet{DeneckerMT00} showed that the $\Psi_\PP$-well-founded fixpoint of $T_\PP$ is the well-founded model of $\PP$ \cite{GelderRS91} and that $\Psi_\PP$-stable fixpoints are exactly the stable models of $\PP$ \cite{iclp/GelfondL88}.

\paragraph{Parametrised Logic Programs}
We briefly recall the parametrised well-founded semantics. This semantics has been implicitly present in the literature for a long time, by assigning a meaning to an \emph{intensional} database. 
We follow the formalisation
by \citet{lpnmr/DeneckerV07}.
For \emph{parametrised logic programs}, the alphabet $\voc$ is partitioned into a set $\voc_p$ of parameter symbols and a set $\voc_d$ of defined symbols. 
Only defined symbols occur in heads of rules. 
Given a $\voc_p$-interpretation $\struct$, $\PP$ defines an immediate consequence operator $T_\PP^\struct:2^{\voc_d}\to 2^{\voc_d}$ equal to $T_\PP$ except that the value of atoms in $\voc_p$ is fixed to their value in $\struct$. 
Similarly, Fitting's immediate consequence operator $\Psi_\PP^\struct$ induces an operator on $(2^{\voc_d})^2$.
$J$ is a \emph{model}\footnote{Note that this definition of model differs from the traditional definition of model of a logic program. To emphasise this difference, we use $J\modelswfm\PP$ to refer to the parametrised well-founded semantics and $J \models\theory$ for the satisfaction relation of propositional logic.} of $\PP$  under the parametrised well-founded semantics (denoted $J\modelswfm \PP$) if $J\cap \voc_d$ is the $\Psi_\PP^{J\cap \voc_p}$-well-founded fixpoint of $T_\PP^{J\cap \voc_p}$. 
By adding a probability distribution over the parameter symbols, we obtain the ProbLog language~\cite{tplp/FierensBRSGTJR15}.



\section{Algebraical Theory}\label{sec:algebraical}
In this section we develop the algebraical foundations of our techniques. 
We follow the intuitions presented in the introduction: we define one operator that ``summarises'' an entire family operators (these will be immediate consequence operators for different interpretations of the parameter symbols). 
We study the relationship between the well-founded fixpoint of the summarising operator and the original operators. 
Before formally introducing \emph{parametrisations}, we focus on a simpler situation: we show that surjective lattice morphisms preserve the well-founded fixpoint.

\subsection{Surjective Lattice Morphisms}

%

\thmwithproof{def:respects:O}{Definition-Proposition}{defprop}{
 Let $O: L\to L$ be an operator and $f:L\to K$ a lattice morphism.
We say that $O$ \emph{respects} $f$ if for every $x,y\in L$ with $f(x)=f(y)$, it holds that $f(O(x)) = f(O(y))$.

If $f$ is surjective and 
$O$ respects $f$, then there exists a unique operator $O_f:K\to K$ with $O_f\circ f = f\circ O$, which we call the \emph{projection of $O$ on $K$}.
}{
We prove the existence and uniqueness of $O_f$. 

Choose $x\in K$. Since $f$ is surjective, there is a $x'\in L$ with $f(x')=x$. We know that $O_f$ must map $x$ to $f(O(x'))$, hence uniqueness follows. 
	Furthermore, this mapping is well-defined (independent of the choice of $x'$) since $O$ respects $f$.}

If $f:L\to K$ is a lattice morphism, $f^2:L^2\to K^2: (x,y)\mapsto (f(x),f(y))$ is a lattice morphism from the bilattice $L^2$ to the bilattice $K^2$. 

\begin{definition}
Let $A: L^2\to L^2$ be an approximator and $f:L\to K$ a lattice morphism.
We say that $A$ \emph{respects} $f$ if $A$ respects $f^2$ in the sense of Definition \ref{def:respects:O}. 
Furthermore, if $f$ is surjective, we define the \emph{projection} of $A$ on $K$ as the unique operator $A_f:K^2\to K^2$ with $A_f\circ f^2=f^2\circ A$.
\end{definition}

\begin{wrapfigure}{R}{0.45\textwidth}
\centering
 $
 \xymatrix@C=6mm@R=6mm{
 \langle L,\leq\rangle \ar@{->>}[r]^f \ar@{~>}[d] \POS!R(-.7)\ar@(ld,lu)^O    & \langle K,\leq \rangle \ar@{~>}[d]  \POS!R(.7)\ar@(rd,ru)_{O_f}\\
 \langle L^2,\leqp\rangle\ar@{->>}[r]^{f^2}  \POS!R(-.7)\ar@(ld,lu)^A  &  \langle K^2,\leqp \rangle \POS!R(.7)\ar@(rd,ru)_{A_f}
} 
$
\caption{Overview of the operators}\label{fig:operators}
\vspace{-5pt}
\end{wrapfigure}
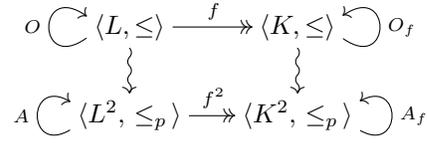
Below, we assume that $f:L\to K$ is a surjective lattice morphism, that 
	$O: L\to L$ is an operator and $A: L^2 \to L^2$ an approximator of $O$ such that both $O$ and $A$ respect $f$ (see Figure \ref{fig:operators}). 
Intuitively elements of $L$ can be thought of as symbolic representations of interpretations, while the elements of $K$ are classical interpretations.
%
%

The following proposition explicates the relationship between well-founded inductions in $L$ and in $K$. This proposition immediately leads to a relationship between the $A$-well-founded model of $O$ and the $A_f$-well-founded model of $O_f$. 
\inproofappendix{
\begin{proposition}\label{prop:refinement}
	If $(x',y')$ is an $A$-refinement of $(x,y)$, then $(f(x'),f(y'))$ is an $A_f$-refinement of $(f(x),f(y))$.
\end{proposition}
\begin{proof}
\ignore{
OUDE VERSIE BEWIJS MET OOK DE OMGEKEERDE RICHTING IN:
1. First suppose $(x',y')$ is an application $A$-refinement of $(x,y)$. Thus \[(x,y)\leqp (x',y')\leqp A(x,y).\] From the fact that $f_i$ is a lattice morphism, it follows that $f_i^2(x,y)\leqp f_i^2(x',y')\leqp f_i^2(A(x,y))$. From the fact that $f_i$ respects $A$, we then find $f_i^2(x,y)\leqp f_i^2(x',y')\leqp A_{f_i}(f_i^2(x,y))$, hence indeed $f_i^2(x',y')$ is an application $A_{f_i}$-refinement of $f_i^2(x,y)$. 

For the other direction, if for every $i$, it holds that $f_i^2(x,y)\leqp f_i^2(x',y')\leqp A_{f_i}(f_i^2(x,y))$, then also for every $i$, it holds that $f_i^2(x,y)\leqp f_i^2(x',y')\leqp f_i^2(A(x,y))$. Since $L$ is a parametrisation of $K$, it also holds that $(x,y)\leqp (x',y')\leqp A(x,y)$.

2. The second point is analogous to the first. First suppose $(x',y')$ is an unfoundedness $A$-refinement of $(x,y)$. Thus $x'=x$ and $A(x,y')_2\leq y'\leq y$. Then also $f_i(x')=f_i(x)$ and $f_i(A(x,y')_2)\leq f_i(y')\leq f_i(y)$, thus $A_{f_i}(f_i(x),f_i(y'))_2\leq f_i(y')\leq f_i(y)$ and the result follows. 

For the reverse direction, if for every $i$, it holds that $f_i(x')=f_i(x)$ and that $A_{f_i}(f(x),f(y'))_2\leq f(y')\leq f(y)$. Then also for every $i$, $f_i(x')=f_i(x)$ and $f(A(x,y')_2)\leq f(y')\leq f(y)$. By the fact that $L$ is a parametrisation of $K$, we find that $x'=x$ and $A(x,y')_2\leq y'\leq y$.}

1. First suppose $(x',y')$ is an application $A$-refinement of $(x,y)$. Thus \[(x,y)\leqp (x',y')\leqp A(x,y).\] From the fact that $f$ is a lattice morphism, it follows that \[f^2(x,y)\leqp f^2(x',y')\leqp f^2(A(x,y)).\] From the fact that $f$ respects $A$, we then find \[f^2(x,y)\leqp f^2(x',y')\leqp A_{f}(f^2(x,y)),\] hence $f^2(x',y')$ is an application $A_{f}$-refinement of $f^2(x,y)$. 

\ 

\noindent
2. The second direction is analogous to the first. Suppose $(x',y')$ is an unfoundedness $A$-refinement of $(x,y)$. Thus $x'=x$ and \[A(x,y')_2\leq y'\leq y.\] Then also $f(x')=f(x)$ and \[f(A(x,y')_2)\leq f(y')\leq f(y),\] thus \[A_{f}(f(x),f(y'))_2\leq f(y')\leq f(y)\] and the result follows. 
}

\inproofappendix{\begin{lemma}\label{lem:monotone:mapping}
If $O$ and $O_f$ are monotone, then $f(\lfp(O)) = \lfp(O_{f})$. 
\end{lemma}
\begin{proof}
The least fixpoint of $O$ is the limit of the sequence $\bot \to O(\bot) \to O(O(\bot))\to \dots$. 
It follows immediately from the definition of $O_f$ that for every ordinal $n$, $f(O^n(\bot)) = O_f^n(f(\bot)) = O_f^n(\bot_K)$, hence the result follows. 
\end{proof}}

\thmwithproof{prop:inductions}{Proposition}{proposition}{
%
	If $(x_j,y_j)_{j\leq \alpha}$ is a well-founded induction of $A$, then $(f(x_j),f(y_j))_{j\leq \alpha}$ is a well-founded induction of~$A_{f}$. If $(x_j,y_j)_{j\leq \alpha}$ is terminal, then so is $(f(x_j),f(y_j))_{j\leq \alpha}$. 
}{
	The first claim follows directly (by induction) from Proposition \ref{prop:refinement}.

	For the second claim, all that is left to show is that if there are no strict $A$-refinements of $(x_\alpha,y_\alpha)$, then there are also no strict $A_{f}$-refinements of $(f(x_\alpha),f(y_\alpha))$. 
 
 First of all, since $(x_\alpha,y_\alpha)$ is a fixpoint of $A$, it also follows for every $i$ that  $A_{f}(f(x_\alpha),f(y_\alpha)) = f^2(A(x_\alpha,y_\alpha))= (f(x_\alpha),f(y_\alpha))$. Thus, there are no strict application refinements of $A_{f}$ either. 
 
 Since there are no unfoundedness refinements of $(x_\alpha,y_\alpha)$, Proposition \ref{prop:biggestufs} yields that $y_\alpha =\lfp S_A^x$.
 It is easy to see that for every $i$, the operator $f\circ S_A^x = S_{A_{f}}^{f(x)}\circ f$. Hence, Lemma \ref{lem:monotone:mapping} (for the operator $S_A^x$) guarantees that $f(y_\alpha) = f(\lfp S_A^x) = \lfp S_{A_{f}}^{f(x)}$. Thus, using Proposition \ref{prop:biggestufs}  we find that there is no strict unfoundedness refinement of $(f(x_\alpha),f(y_\alpha))$. 

}


\thmwithproof{thm:main}{Theorem}{theorem}{
 If $(x,y)$ is the $A$-well-founded fixpoint of $O$, then, $(f(x),f(y))$ is the $A_{f}$-well-founded fixpoint of $O_{f}$.
}{
 Follows immediately from Proposition \ref{prop:inductions}.
}

\subsection{Parametrisations}\label{ssec:param}

\begin{definition}
	Let $L$ and $K$ be lattices. Suppose $(f_i:L\to K)_{i\in I}$ is a family of surjective lattice morphisms. 
We call $L$ a \emph{parametrisation} of $K$ (through $(f_i)_{i\in I}$) if for every $x,y\in L$ it holds that $x\leq y$ if and only if for every $i\in I$, $f_i(x)\leq f_i(y)$.
\end{definition}


A parametrisation $L$ of a lattice $K$ can be used to ``summarise'' multiple operators (the $O_{f_i}$) on $K$ by means of a single operator $O$ on $L$ which abstracts away certain details. 
In the next section, we use this to compute a symbolic representation of the parametrised well-founded model.

\thmwithproof{thm:twoval}{Theorem}{theorem}{
	Suppose $L$ is a parametrisation of $K$ through $(f_i)_{i\in I}$. Let $O: L\to L$ be an operator and $A$ an approximator of $O$ such that both $O$ and $A$ respect each of the $f_i$. If $(x,y)$ is the $A$-well-founded fixpoint of $O$, the following hold. 
	\begin{itemize}
	 \item[1.] For each $i$, $(f_i(x),f_i(y))$ is the $A_{f_i}$-well-founded fixpoint of $O_{f_i}$. 
	 \item[2.] If the $A_{f_i}$-well-founded fixpoint of $O_{f_i}$ is exact for every $i$, then so is the $A$-well-founded fixpoint of $O$. 
	\end{itemize}
}{
The first point immediately follows from Theorem \ref{thm:main}. 

Using the first point, we find that if the $A_{f_i}$-well-founded fixpoint of $O_{f_i}$ is exact for every $i$, then $f_i(x)= f_i(y)$ for every $i$. 
Hence the definition of parametrisation guarantees that $x=y $ as well, i.e., the $A$-well-founded fixpoint of $O$ is indeed exact. 
}

\section{Operator-Based Knowledge Compilation}\label{sec:kp-lp}

We assume throughout this section that \PP refers to a parametrised logic program with parameters $\voc_p$ and defined symbols $\voc_d$.
In order to apply our theory to logic programming, we will define an operator (and approximator) that summarises the immediate consequence operators of $\PP$ for all $\voc_p$-interpretations. 

Partial interpretations map defined atoms to a tuple $(t,p)$ of two-valued truth values. 
We generalise this type of interpretations: we want (partial) interpretations to be parametrised in terms of the parameters of the logic program. 
Instead of assigning a tuple $(t,p)$ of Boolean values to each atom, we will hence assign a tuple of two propositional formulas over $\voc_p$ to each atom in $\voc_d$. 

In order to avoid redundancies, we work \emph{modulo equivalence}. 
Let $\LL_{\voc_p}$ be the language of all propositional formulas over vocabulary $\voc_p$. 
If $\varphi$ is a propositional formula, we use $\bar\varphi$ to denote the equivalence class of $\varphi$, i.e., the set of propositional formulas equivalent to $\varphi$.\footnote{Notice that $\bar a$ is \emph{not} the negation of an atom $a$. We use $\lnot a$ for the negation of $a$.}
Let $\plat$ be the set of equivalence classes of elements in $\LL_{\voc_p}$. 
We define an order $\leq_{\plat}$ on $\plat$ as follows: $\bar\varphi\leq_{\plat} \bar\psi$ if $\varphi$ entails $\psi$ (in standard propositional logic). 
This order is well-defined (independent of the choice of representatives $\varphi$ and $\psi$); 
with this order, $\plat$ is a complete lattice.
Boolean operations on \plat are defined by applying them to representatives.

\begin{definition}\label{def:spstruct}
A \emph{symbolic interpretation} of $\voc_d$ in terms of $\voc_p$ is a mapping $\voc_d\to \plat$. 
The \emph{symbolic interpretation lattice} $\dplat$ is the set of all symbolic interpretations of $\voc_d$ in terms of $\voc_p$. The order $\leq$ on $\dplat$ is the pointwise extension of $\leq_{\plat}$.
A \emph{partial symbolic interpretation} is an element of the bilattice $(t,p)\in(\dplat)^2$ such that $t\leq p$.%
\end{definition}
The condition $t\leq p$ in Definition \ref{def:spstruct} excludes inconsistent interpretations. 
If $\voc_p$ is the empty vocabulary (i.e., if \PP has no parameters), then the lattice $\plat$ is $\{\bar\lfalse,\bar\ltrue\}$ with order $\bar\lfalse\leq\bar\ltrue$.
Hence, in this case, a (partial) symbolic interpretation is ``just'' a (partial) interpretation. 
As with classical interpretations, we often identify a symbolic interpretation \sstruct with the partial symbolic interpretation $(\sstruct,\sstruct)$.

Intuitively, a (partial) symbolic interpretation summarises many different classical (partial) interpretations; when we instantiate such as (partial) symbolic interpretation with a $\voc_p$-interpretation, we obtain a unique (partial) $\voc_d$-interpretation. 
The following definition formalises this intuition. 

\begin{definition}\label{def:concretisation}
 If $\spstruct=(\sstruct_t,\sstruct_p)$ is a partial symbolic interpretation and $\struct$ is a $\voc_p$-interpretation, the \emph{concretisation} of \spstruct by \struct is the partial interpretation $\spstruct^\struct$ such that for every symbol $a\in\voc_d$ with $\sstruct_t(a) = \overline{\varphi_t}$ and $\sstruct_p(a)=\overline{\varphi_p}$, it holds that   $\spstruct^\struct(a)= (\varphi_t^\struct,\varphi_p^\struct)$.
\end{definition}
The above concept is well-defined (independent of the choice of representatives $\varphi_t$ en $\varphi_p$).
A symbolic interpretation can thus be seen as a mapping from $\voc_p$-interpretations to $\voc_d$-interpretations. 
This kind of mapping is of particular interest, since the parametrised well-founded semantics induces a similar mapping: it associates with every $\voc_p$-interpretation a $\voc_d$-interpretation, namely the $\Psi_\PP^I$-well-founded model of $T_\PP^I$. It is this relationship between $\voc_p$- and $\voc_d$-interpretations that we wish to capture in propositional logic. 
Furthermore, as explained below, it is easy to translate a symbolic interpretation into propositional logic. 

\begin{definition}
 Let \sstruct be a symbolic interpretation and $\psi_p$ a representative of $\sstruct(p)$ for each $p\in \voc_d$. 
 We call a propositional theory \theory \emph{a theory of} \sstruct if it is equivalent to 
 $\bigand_{p\in\voc_d}p\lequiv \psi_p.$
%
\end{definition}
All theories of \sstruct are equivalent. 
We sometimes abuse notation and refer to \emph{the} theory of \sstruct, denoted $\Th(\sstruct)$, to refer to any theory from this class. 
The goal now is to find a symbolic interpretation $\sstruct$ such that $\Th(\sstruct)$ is equivalent to~$\PP$. 
Our choice of representatives will depend on the target language of the compilation.   

The value of a propositional formula $\varphi$ in a partial interpretation \pstruct is an element of $\{\ltrue,\lfalse,\lunkn\}$ (or, a tuple of two Booleans) obtained by  standard three-valued valuation. This can easily be extended to symbolic interpretations, where the value of a formula in a (partial) symbolic interpretation is a tuple of two $\voc_p$ formulas. 

\begin{definition}\label{def:eval:partial}
 Let $\varphi$ be a $\voc$-formula and $\spstruct=(\sstruct_t,\sstruct_p)$ a partial symbolic interpretation. The value of $\varphi$ in \spstruct is a tuple $(\varphi_{t},\varphi_{p})\in \plat^2$ defined inductively as follows:
 \begin{itemize}
 \item $p^{(\sstruct_t,\sstruct_p)} = (\bar p,\bar p)$ if $p\in \voc_p$ and  $p^{(\sstruct_t,\sstruct_p)} = (\sstruct_t(p), \sstruct_p(p))$ if $p\in \voc_d$,
 \item $(\psi\land \xi)^{(\sstruct_t,\sstruct_p)} = (\overline{\psi_t\land \xi_t}, \overline{\psi_p\land \xi_p} )$ if $\psi^{(\sstruct_t,\sstruct_p)}= (\overline{\psi_t},\overline{\psi_p})$ and $\xi^{(\sstruct_t,\sstruct_p)}= (\overline{\xi_t},\overline{\xi_p})$
 \item $(\psi\lor  \xi)^{(\sstruct_t,\sstruct_p)} = (\overline{\psi_t\lor  \xi_t}, \overline{\psi_p\lor  \xi_p} )$ if $\psi^{(\sstruct_t,\sstruct_p)}= (\overline{\psi_t},\overline{\psi_p})$ and $\xi^{(\sstruct_t,\sstruct_p)}= (\overline{\xi_t},\overline{\xi_p})$
  \item $(\lnot \psi)^{(\sstruct_t,\sstruct_p)} = (\overline{\lnot \psi_p},\overline{\lnot \psi_t})$ if $\psi^{(\sstruct_t,\sstruct_p)}= (\overline{\psi_t},\overline{\psi_p})$.
 \end{itemize}
\end{definition}

Evaluation of formulas has some nice properties. It commutes with concretisation (Proposition \ref{prop:lattice:param}) and induces a parametrisation (Proposition \ref{prop:dplat:param}).

\thmwithproof{prop:lattice:param}{Proposition}{proposition}{
 For every formula $\varphi$ over \voc, $\spstruct\in(\dplat)^2$ and $\struct\in 2^{\voc_p}$, it holds that $\varphi^{\spstruct^\struct} = (\varphi^\spstruct)^\struct$.}{
 Trivial.
}
\thmwithproof{prop:dplat:param}{Proposition}{proposition}{
The lattice \dplat is a parametrisation of $2^{\voc_d}$ through the mappings $(\proj_\struct:\dplat\to 2^{\voc_d}:\sstruct\mapsto \sstruct^\struct)_{\struct\in 2^{\voc_p}}$.
}{
 It is clear that the mappings $\proj_\struct$ are lattice morphisms since evaluation of propositional formulas commutes with Boolean operations. 
 Now, for $\sstruct,\sstruct'\in\dplat$, it holds that $\sstruct\leq \sstruct'$ if and only if for every atom $p\in {\voc_d}$, $\sstruct(p)$ entails $\sstruct'(p)$. 
 This is equivalent to the condition that for every $p\in {\voc_d}$ and every interpretation $\struct\in 2^{\voc_d}$, $\sstruct(p)^\struct\leq\sstruct'(p)^\struct$, i.e., with the fact that for every $I$, $\proj_\struct(\sstruct)\leq \proj_\struct(\sstruct')$ which is what we needed to show.
}

Recall from Section~\ref{ss:lp} that $\varphi_p$ is the disjunction of all bodies of rules defining $p$; using this we can generalise both $T_\PP$ and $\Psi_\PP$ to a symbolic setting. 

\begin{definition}
The \emph{partial parametrised immediate consequence operator} $\ppimcons:(\dplat)^2\to(\dplat)^2$ is defined by
$
 \ppimcons(\spstruct)(p)=\varphi_p^\spstruct$   for every $p\in\voc_d$.
 
The \emph{parametrised immediate consequence operator} is the operator $\pimcons: \dplat\to\dplat$ that maps $\sstruct$ to $\pimcons(\sstruct)$, where
$\pimcons(\sstruct)(p) = \varphi_p^\sstruct$ for each $p\in \voc_d$.
\end{definition}
It deserves to be noticed that the operator $\pimcons$ almost coincides with the operator  $\mathcal{T}_{c_\PP}$  defined by \citet{Jonas} (the only difference is that we work modulo equivalence). 
The following proposition, which follows easily from our algebraical theory, shows correctness of the methods developed by \citet{Jonas}. 

 
 \thmwithproof{thm:posprog}{Theorem}{theorem}{
  If $\PP$ is a positive logic program, then \pimcons is monotone. 
  For every $\voc$-interpretation $I$, it then holds that $I\modelswfm \PP$ if and only if  $I\models \Th(\lfp(\pimcons))$.
}{  Follows immediately from the definition of the parametrised well-founded semantics combined with Lemma \ref{lem:monotone:mapping}. 
 }

\thmwithproof{thm:approx}{Theorem}{theorem}{
For any parametrised logic program \PP, the following hold:
\begin{itemize}
	\item[1.] $\ppimcons$ is an approximator of $\pimcons$.
	\item[2.] For every $\voc_p$-structure $\struct$, it holds that $\Psi_{\PP}^{\struct}  \circ \proj_\struct^2 = \proj_\struct^2\circ \ppimcons$.
\end{itemize}

}{
1. It follows immediately from the definitions that for exact interpretations $\spstruct=(\sstruct,\sstruct)$, $\ppimcons$ coincides with $\pimcons$.
 $\leqp$-monotonicity follows directly from the definition of evaluation of formulas (Definition \ref{def:eval:partial}). 

\noindent 2. We find that for every $\spstruct\in (\dplat)^2$ and every $p\in 2^{\voc_d}$, 
 \begin{align*}
 \Psi_{\PP}^{\struct}  (\proj_\struct^2(\spstruct))(p) &= \Psi_{\PP}^{\struct} (\spstruct^\struct) (p)\\
 &=\varphi_p^{\spstruct^\struct}\\
 &=(\varphi_p^\spstruct)^\struct\\
 &= (\ppimcons(\spstruct)(p))^\struct\\
 &= \proj_\struct^2( \ppimcons(\spstruct)(p)),
 \end{align*}
which indeed proves our claim.
}

\begin{definition}
Let \PP be any parametrised logic program. The \emph{parametrised well-founded model} of \PP is the $\ppimcons$-well-founded fixpoint of $\pimcons$.
\end{definition}

Applying Theorem \ref{thm:main}, combined with Proposition \ref{prop:lattice:param} and Theorem \ref{thm:approx} yields:

\begin{theorem}\label{thm:wfm:param:lp}
 If the parametrised well-founded model of \PP is exact, i.e., of the form $(\sstruct,\sstruct)$ for some symbolic interpretation \sstruct, then for every $\voc$-interpretation $\struct$,   it holds that $I\modelswfm \PP$ if and only if  $I\models \Th(\sstruct)$.
\end{theorem}

\ignore{
\begin{example}\label{ex:bigexample}
	We illustrate the various concepts introduced above. Let $\PP_r$ be the logic program
%
%
\[\left\{
\begin{array}{lll}
 r(a,a). &r(b,c) \lrule e(b,c). \\
 r(a,b)\lrule e(a,b). &r(a,b)\lrule e(a,a)\land r(a,b) \\ 
	r(a,c)\lrule e(a,c). &r(a,c)\lrule e(a,a)\land r(a,c). &r(a,c)\lrule e(a,b)\land r(b,c)
 \end{array}
\right\}\]

with parameters $e(\cdot,\cdot)$ and defined symbols $r(\cdot,\cdot)$. 
In this case, the parametrised well-founded model of $\PP_r$ is the symbolic interpretation $\sstruct_r: \voc_d\to \plat:
	r(a,a) \mapsto \overline{\ltrue}, \quad
	r(a,b)\mapsto  \overline{e(a,b)}, \quad
	r(b,c) \mapsto \overline{ e(b,c)}, \quad
	r(a,c)\mapsto \overline{ e(a,c) \lor (e(a,b)\land e(b,c))}.$

Notice that $\Th(\sstruct_r)$ is equivalent to $\PP_r$, in the sense that $J\models \Th(\sstruct_r)$ if and only if $J\modelswfm \PP_r$. 
For example, let $\struct$ be the $\voc_p$-interpretation $\{e(a,b), e(b,c)\}$. We know that the well-founded model of $T_{\PP_r}^I$ is $I':=\{r(a,a), r(a,b), r(b,c), r(a,c)\}$; this equals $\sstruct_r^\struct$ and $I\cup I'$ is indeed a model of $\Th(\sstruct_r)$. 

Since $\PP_r$ is positive, $\mathcal{T}_{\PP_r}$ is monotone and its least fixpoint can be computed by 
iteratively applying the operator $\mathcal{T}_{\PP_r}$ starting from the smallest symbolic interpretation; this yields the following sequence:
\begin{align*}
	\bot&: &r(a,a) &\mapsto  \bar\lfalse,&  r(a,b) &\mapsto \bar\lfalse,& r(b,c) &\mapsto  \bar\lfalse,&   r(a,c) &\mapsto  \bar\lfalse, \\
	\mathcal{T}_{\PP_r}(\bot)&: &r(a,a) &\mapsto  \bar\ltrue,& r(a,b) &\mapsto \overline{ e(a,b)},& r(b,c) &\mapsto  \overline{ e(b,c)}, &  r(a,c) &\mapsto  \overline{ e(a,c)},  \\
	\mathcal{T}_{\PP_r}^2(\bot) &= &\sstruct_r. \qquad
\end{align*}

%
Also, $\mathcal{T}_{\PP_r}^3(\bot)=\mathcal{T}_{\PP_r}^2(\bot)$ and we find that $\lfp(\mathcal{T}_{\PP_r})= \sstruct_r$.
\end{example}
}

\newcommand{\smokes}{\m{\mathit{smokes}}}
\newcommand{\stress}{\m{\mathit{stress}}}
\newcommand{\friends}{\m{\mathit{fr}}}
\begin{example}\label{ex:smokers}
 We illustrate the various concepts introduced above on the smokers problem, a popular problem in probabilistic logic programming. 
 Consider a group of people. A person of this group smokes if he is stressed, or if he is friends with a smoker. This results in the following logic program $\PP_s$ with a domain of three people $\{a,b,c\}$:
 \[\left\{
\begin{array}{lll}
 \forall X: \smokes(X)\lrule \stress(X)  \\
 \forall X, Y: \smokes(X) \lrule \friends(X,Y) \land \smokes(Y)
 \end{array}
\right\}\]
This program has parameters $\stress(\cdot)$ and $\friends(\cdot,\cdot)$ and defined symbols $\smokes(\cdot)$. 
The parametrised well-founded model of $\PP_s$ is the symbolic interpretation $\sstruct_s: \voc_d\to \plat:$ such that 
\begin{align*}
 \sstruct_s(\smokes(a)) = &\overline{\stress(a)\lor (\stress(b) \land \friends(a,b)) \lor (\stress(c)\land \friends(a,c))} \\
			 &\overline{\lor(\stress(c) \land \friends(b,c)\land \friends(a,b))}\\
			 &\overline{\lor(\stress(b) \land \friends(c,b)\land \friends(a,c)) }
\end{align*}
and symmetrical equations hold for $\smokes(b)$ and $\smokes(c)$.

Notice that $\Th(\sstruct_s)$ is equivalent to $\PP_s$, in the sense that $J\models \Th(\sstruct_s)$ if and only if $J\modelswfm \PP_s$. 
For example, let $\struct$ be the $\voc_p$-interpretation $\{\stress(a),\friends(b,a)\}$. We know that the $\Psi_{\PP_s}^I$-well-founded fixpoint of $T_{\PP_r}^I$ is $I':=\{\smokes(a),\smokes(b)\}$; this equals $\sstruct_s^\struct$ and $I\cup I'$ is indeed a model of $\Th(\sstruct_s)$. 

Since $\PP_s$ is positive, $\mathcal{T}_{\PP_s}$ is monotone and its least fixpoint can be computed by 
iteratively applying the operator $\mathcal{T}_{\PP_s}$ starting from the smallest symbolic interpretation; this yields the following sequence (only the value of $\smokes(a)$ is explicated; for $\smokes(b)$ and $\smokes(c)$, similar equations hold):
\begin{align*}
  \bot &: &\smokes(a)&\mapsto \bar\lfalse\\
  \mathcal{T}_{\PP_s}(\bot)&: &\smokes(a)&\mapsto \overline{\stress(a)}\\
  \mathcal{T}_{\PP_s}^2(\bot)&: &\smokes(a)&\mapsto \overline{\stress(a)\lor (\stress(b) \land \friends(a,b)) \lor (\stress(c)\land \friends(a,c))}\\
  \mathcal{T}_{\PP_s}^3(\bot)&= &\sstruct_s. \qquad
\end{align*}
In Figure \ref{fig:smoker}, a circuit representation of $\Th(\sstruct_s)$ is depicted. In this circuit, the different layers correspond to different steps in the computation of the parametrised well-founded model of $\PP_s$. Figure \ref{fig:smoker} essentially contains proofs of atoms $\smokes(\cdot)$; this illustrates that the compiled theory can be used for example for abduction. 
\end{example}

\begin{figure}[htb]
  \centering
  \includegraphics{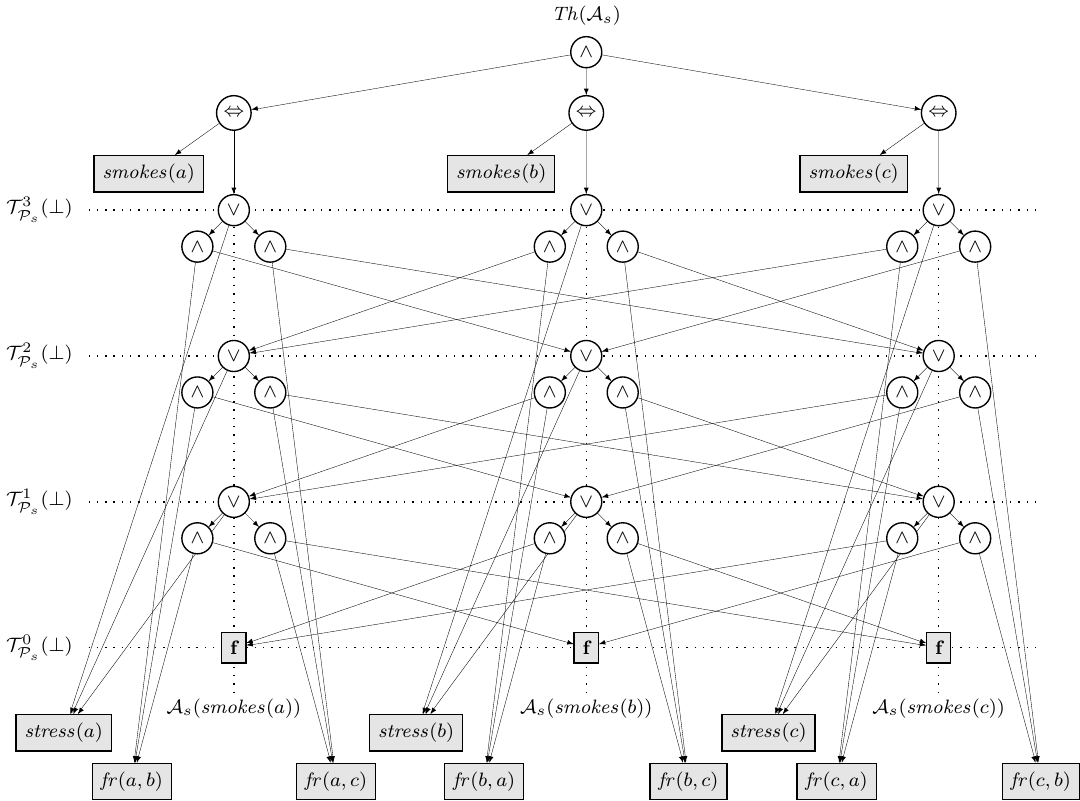}
  \caption{A circuit representation of the smokers theory $\Th(\sstruct_s)$ and the different steps in the computation of $\mathcal{T}_{\PP_s}$.}
  \label{fig:smoker}
\end{figure}

For general logic programs, $\pimcons$ is not guaranteed to be monotone and hence the parametrised well-founded model cannot be computed by iteratively applying \pimcons. Luckily, well-founded inductions provide us with a constructive way to compute it.

\newcommand{\pimconsw}{\m{\mathcal{T}_{\PP_w}}}
\newcommand{\ppimconsw}{\m{\varPsi_{\PP_w}}}
\newcommand{\pimconss}{\m{\mathcal{T}_{\PP_s}}}
\newcommand{\ppimconss}{\m{\varPsi_{\PP_s}}}

\newcommand{\turns}{\m{\mathit{turns}}}
\newcommand{\button}{\m{\mathit{button}}}

\begin{example}\label{ex:gear}
 Consider a dynamic domain in which two gear wheels are connected. Both wheels can be activated by an external force; since they are connected, whenever one wheel turns, so does the other. Both wheels are connected to a button. If an operator hits the button associated to some gear wheel, this means that he intends the state of the wheel to change (if a wheel was turning, its external force is turned off, if the wheel was standing still, its external force is activated). If the operator does not hit the button, the external force is set to the current state of the wheel. 
 Initially, both external forces are inactive. 
  This situation (limited to two time points) is modelled in the following logic program $\PP_w$ ($\turns_i(T)$ means that wheel $i$ is turning at time point $T$ and $\button_i(T)$ means that the button of wheel $i$ is pressed at time $T$):
  \[\left\{
\begin{array}{ll}
\turns_1(0) \lrule \turns_2(0)&
\turns_2(0) \lrule \turns_1(0)\\
\turns_1(1) \lrule \turns_2(1)&
\turns_2(1) \lrule \turns_1(1)\\
\turns_1(1) \lrule \turns_1(0) \land \lnot \button_1(0)&
\turns_2(1) \lrule \turns_2(0) \land \lnot \button_2(0)\\
\turns_1(1) \lrule \lnot \turns_1(0) \land \button_1(0)&
\turns_2(1) \lrule \lnot \turns_2(0) \land \button_2(0)
 \end{array}
\right\}\]
This logic program has defined symbols $\turns_\cdot(\cdot)$ and parameters $\button_\cdot(\cdot)$. 
The parametrised well-founded model of $\PP_w$ is computed by a well-founded induction of $\ppimconsw$. 
We start from the least precise partial symbolic interpretation, i.e., $\spstruct_0$ that maps every $\turns_\cdot(\cdot)$ to $(\bar\lfalse,\bar\ltrue)$. 
Since $\spstruct_0$ is a fixpoint of $\ppimconsw$, the only possible type of refinement is unfoundedness refinement, resulting in $\spstruct_1$  that maps
\begin{align*}
 \turns_1(0) &\mapsto (\bar\lfalse,\bar\lfalse) & \turns_2(0) &\mapsto (\bar\lfalse,\bar\lfalse) \\
 \turns_1(1) &\mapsto (\bar\lfalse,\bar\ltrue)  & \turns_2(1) &\mapsto (\bar\lfalse,\bar\ltrue)
\end{align*}
Application refinement then results in the partial symbolic interpretation $\spstruct_2 = \ppimconsw(\spstruct_1)$ that maps
\begin{align*}
 \turns_1(0) &\mapsto (\bar\lfalse,\bar\lfalse) & \turns_2(0) &\mapsto (\bar\lfalse,\bar\lfalse) \\
 \turns_1(1) &\mapsto (\overline{\button_1(0)},\bar\ltrue)  & \turns_2(1) &\mapsto (\overline{\button_2(0)},\bar\ltrue)
\end{align*}
Another application refinement then results in the partial symbolic interpretation $\spstruct_3 = \ppimconsw(\spstruct_2)$  that maps
\begin{align*}
 \turns_1(0) &\mapsto (\bar\lfalse,\bar\lfalse) & \turns_2(0) &\mapsto (\bar\lfalse,\bar\lfalse) \\
 \turns_1(1) &\mapsto (\overline{\button_1(0)\lor \button_2(0)},\bar\ltrue)  & \turns_2(1) &\mapsto (\overline{\button_2(0)\lor \button_1(0)},\bar\ltrue)
\end{align*}
Finally, one last unfoundedness refinement results in the symbolic interpretation $\sstruct_w$  that maps
\begin{align*}
 \turns_1(0)& \mapsto \bar \lfalse & \turns_2(0)& \mapsto \bar \lfalse \\
 \turns_1(1)&\mapsto \overline{\button_1(0)\lor \button_2(0)} & \turns_2(1)&\mapsto\overline{\button_1(0)\lor \button_2(0)}  
\end{align*}


In Figure A.1 in online Appendix A, a circuit representation of $\Th(\sstruct_w)$ is depicted. In this circuit, the different layers correspond to the evolution of the lower bound in different steps in the computation of the parametrised well-founded model of $\PP_w$ (unfoundedness refinements are not visualised). In Figure A.2, the circuit for this examples with time ranging from $0$ to $2$ is depicted. 
\end{example}


\begin{example}[Example \ref{ex:smokers} continued]\label{ex:smokers:continued}
 Well-founded inductions also work for positive logic programs. Let $\spstruct_0$ denote the least precise partial interpretation. Since $\PP_s$ is positive, it holds for every $i$ and $X$ that 
 \[\ppimconss^i(\spstruct_0)(\smokes(X)) = (\pimconss^i(\bot)(\smokes(X)), \bar\ltrue).\]
 Hence, repeated application refinements yield the partial symbolic interpretation $(\sstruct_s,\top)$. One final unfoundedness refinement then results in the parametrised well-founded model of $\PP_s$, namely $\sstruct_s$. 
\end{example}

\subsection*{Discussion}

The condition in Theorem \ref{thm:wfm:param:lp} naturally raises the question ``what happens if the parametrised well-founded model is \emph{not} exact?''. First of all, our techniques also work in this setting. Indeed, Theorem \ref{thm:twoval} (1) guarantees that instantiating the  the parametrised well-founded model of \PP with  a $\voc_p$-interpretation $\struct$ results in the $\Psi_\PP^\struct$-well-founded fixpoint of $T_\PP^\struct$. 
\begin{example}\label{ex:nontot}
 Let $\PP_{\mathit{NT}}$ be the following logic program 
 \[\left\{
\begin{array}{lllll}
 a\lrule \lnot b. & b\lrule \lnot a. &
 c\lrule \lnot b& c\lrule e. & d \lrule a \land \lnot c.
 \end{array}
\right\}\]
with parameter symbol $e$ and defined symbols $a,b,c$ and $d$. 
The parametrised well-founded model of $\PP_{\mathit{NT}}$ is then $\spstruct_{\mathit{NT}}$ such that 
\begin{align*}
\spstruct_{\mathit{NT}}(a) &= (\overline \lfalse, \overline \ltrue) & \spstruct_{\mathit{NT}}(b) &= (\overline \lfalse, \overline \ltrue) &
\spstruct_{\mathit{NT}}(c) &= (\overline e, \overline \ltrue) & \spstruct_{\mathit{NT}}(d) &= (\overline \lfalse, \overline {\lnot e}) 
\end{align*}
\end{example}

However, in this text we mainly focus on programs with an exact parametrised well-founded model.  
Corollary \ref{thm:twoval} guarantees that this condition is satisfied for all logic programs in which the standard well-founded model is two-valued. 
This kind of programs is common in applications for deductive databases \cite{jcss/AbiteboulV91} and for representing inductive definitions \mycite{NaturalInductions}. 
Classes that satisfy this condition include monotone and (locally) stratified logic programs \cite{minker88/Przymusinski88}. 

This restriction is typically not satisfied by ASP programs, where stable semantics is used. 
However, it deserves to be stressed that there is a strong relationship between ASP programs and logic programs under the parametrised well-founded semantics. 
Most ASP programs, e.g., those used in ASP competitions, are so-called generate-define-test (GDT) programs.
They consist of three modules. 
A generate module  opens the search space (i.e., it introduces parameter symbols); a define module contains inductive definitions for which well-founded and stable semantics coincide (as argued by \citet{\refto{NaturalInductions}}) and a test module consist of constraints. 
\citet{DeneckerLTV12} have argued that a GDT program is the \emph{monotone conjunction} of its different modules. 
Hence, our technique can be used to compile the \emph{define} part of a GDT program. The example below illustrates that only compiling this part results in an interpretation that captures the meaning of this definition more closely, by preserving more structural information. 
\begin{example}[Example \ref{ex:nontot} continued]
 The first two rules of  $\PP_{\mathit{NT}}$ encode a choice rule for $a$ (or $b$). The define module of this program is the program 
  \[\PP_{\mathit{def}}=\left\{
\begin{array}{llll}
  b\lrule \lnot a. &
 c\lrule \lnot b& c\lrule e. & d \lrule a \land \lnot c.
 \end{array}
\right\}\]
with parameter symbols $a$ and $e$, and defined symbols $b,c$ and $d$. The parametrised well-founded model of $\PP_{\mathit{def}}$ is the symbolic interpretation $\sstruct_{\mathit{def}}$ such that 
\begin{align*}
 \sstruct_{\mathit{def}}(b) &= \overline{\lnot a}     &
 \sstruct_{\mathit{def}}(c) &= \overline{a \lor e}    &
 \sstruct_{\mathit{def}}(d) &= \overline{a \land \lnot (a\lor e)} = \overline{\lfalse}  
\end{align*}
As can be seen, the parametrised well-founded model now contains the information that $d$ is false, independent of the value of the parameter symbols (independent of the choice made in the choice rules in the original example).
\end{example}



\ignore{

\begin{example}[Example \ref{ex:symbolicinter} continued]\label{ex:sstruct}
\end{example}


\begin{example}[Example \ref{ex:sstruct} continued]\label{ex:sstruct:c}
Let $\sstruct'$ be the symbolic interpretation such that $\sstruct'$ interprets all atoms $r(x,y)$ as $\bar\ltrue$. Then $\sstruct\leq\sstruct'$ since every formula entails $\ltrue$.
\end{example}

\begin{example}[Example \ref{ex:symbolicinter} continued]
Again consider the logic program \PP from Example \ref{ex:symbolicinter}. Iteratively applying the operator $\pimcons$ starting from the smallest symbolic interpretation yields the following sequence:
\[
 \begin{array}{c|cccc}
  & \bot & \pimcons(\bot) & \pimcons(\pimcons(\bot))\\
  \hline
  r(a,a) & \bar\lfalse & \bar\ltrue& \bar\ltrue\\
  r(a,b) & \bar\lfalse & \overline{ e(a,b)} & \overline{ e(a,b)} \\
  r(b,c) & \bar\lfalse & \overline{ e(b,c)} &\overline{ e(b,c)} \\
  r(a,c) & \bar\lfalse & \overline{ e(a,c)} &\overline{ e(a,c) \lor (e(a,b)\land e(b,c))}  
 \end{array}
\]
%
Also, $\pimcons^3(\bot)=\pimcons^2(\bot)$; the operator \pimcons is monotone and we find that its least fixpoint is \sstruct.
\end{example}}

\section{Algorithms}

Based on the theory developed in the previous section, we now discuss practical algorithms for exact and approximate knowledge compilation of logic programs.

\subsection{Exact Knowledge Compilation}
\newcommand{\state}{\m{\mathfrak{S}}}
The definition of a well-founded induction provides us with a fixpoint procedure to compute the parametrised well-founded model. 
Our algorithms are parametrised by a language \LL, referred to as the \emph{target language}; this can be any representation of propositional formulas.  
We describe our algorithm, which we call $\compile(\LL)$, as a (non-deterministic) finite-state-machine.  
A \emph{state} $\state$ consists of an assignment of two formulas $\state_t(q)$ and $\state_p(q)$ in $\LL$  (over vocabulary $\voc_p$) to each atom $q\in \voc_d$. 
Hence, a state $\state$ corresponds to the partial symbolic interpretation $\spstruct_\state = (\sstruct_t,\sstruct_p)$ such that for each $q\in \voc_d$, $\sstruct_t(q) = \overline{\state_t(q)}$ and $\sstruct_p(q) = \overline{\state_p(q)}$.  
The \emph{transitions} in our finite-state-machine are exactly those tuples of states $(\state, \state')$ such that $\spstruct_{\state'}$ is a $\ppimcons$-refinement of $\spstruct_{\state}$. 

We further restrict these transitions  to \emph{maximally precise} transitions: \emph{application refinements} that refine $\spstruct$ to $\ppimcons(\spstruct)$ and \emph{unfoundedness refinements} as described in Proposition \ref{prop:biggestufs}. 
Furthermore, we propose to make the resulting finite-state-machine \emph{deterministic} by prioritising application refinements over unfoundedness refinements since they are cheaper, i.e., they only require one application of $\ppimcons$.

%

The final output of $\compile(\LL)$ is a theory $\Th(\sstruct)$ in \LL, where $\sstruct$ is the parametrised well-founded model of \PP.
When $\LL$ denotes Boolean circuits, each application of $\ppimcons$ adds a layer of Boolean gates over the circuits in $\spstruct_s$. When $\LL$ denotes a language with a so-called \textsc{Apply} function~\cite{aaai/BroeckD15} (e.g.,  SDDs), each application of $\ppimcons$ calls \textsc{Apply} to conjoin or disjoin circuits from $\spstruct_s$.

Figure \ref{fig:smoker} contains an example circuit for the smokers problem (Example \ref{ex:smokers}). 
The different layers in the circuit correspond to different steps in a well-founded induction (or the least fixpoint computation). Our algorithm follows the well-founded induction as described in Example \ref{ex:smokers:continued}, by prioritising application refinements over unfoundedness refinements. 
Similarly, our algorithm also follows the well-founded induction from Example \ref{ex:gear}. 
During the execution, circuits to represent the upper and lower bounds are gradually built (layer by layer). 

\inproofappendix{
\begin{lemma}\label{lem:nbrefinements:instantiated}
For every $\voc_p$-interpretation $I$, there are at most $|\voc_d|$ strict refinements in a well-founded induction of $\Psi_\PP^I$. 
\end{lemma}
\begin{proof}
	Every strict refinement should at least change one of the atoms in $\voc_d$ from unknown to either true or false, hence the result follows. 
\end{proof}

\begin{lemma}\label{lem:nbrefinements:parametrised}
 Suppose $(x_i,y_i)_{i\leq \beta}$ is a well-founded induction of $\pimcons$ in which every refinement is maximally precise, i.e., either of the form $(x,y)\to \pimcons(x,y)$ or an unfoundedness refinement satisfying the condition in Proposition \ref{prop:biggestufs}. The following hold:
\begin{itemize}
	\item there are at most $|\voc_d|$ subsequent strict application refinements in $(x_i,y_i)_{i\leq \beta}$, and 
	\item if unfoundedness refinements only happen in $(x_i,y_i)_{i\leq \beta}$ when no application refinement is possible, then there are at most $|\voc_d|$ unfoundedness refinements.
\end{itemize}
\end{lemma}
\begin{proof}
For the first part, we notice that every sequence of maximal application refinements maps (by $\pi_I$) onto a sequence of maximal application refinements of $\Psi_\PP^I$. 
Furthermore, from the proof of Proposition \ref{prop:inductions}, it follows that if a $\pimcons$-refinement is strict, then at least on of the induced $\Psi_\PP^I$-refinements must be strict as well. The result now follows from Lemma \ref{lem:nbrefinements:instantiated}.

The second point is completely similar to the first. There can be at most $|\voc_d|$ strict unfoundedness refinements in any well-founded induction of $\Psi_\PP^I$. Furthermore, the condition in this point guarantees that if for some $I$, an unfoundedness refinement in the induced well-founded induction is not strict, then neither will any later unfoundedness refinements. Hence, the result follows.
\end{proof}
}

\thmwithproof{thm:complexity}{Theorem}{theorem}{
Let $\LL_{\mathit{BC}}$ be the language of Boolean circuits. The following hold:  \textit{(i)}
 $\compile(\LL_{\mathit{BC}})$ has polynomial-time complexity and \textit{(ii)}
   the size of the output circuit of $\compile(\LL_{\mathit{BC}})$  is polynomial in the size of~\PP. 
}{
First, we notice that if we have a circuit representation of $\spstruct$, then the representation of $\ppimcons(\spstruct)$ consists of the same circuit with maximally three added layers since $\varphi_p$ is a DNF for every defined atom $p$ (a layer of negations, one of disjunctions and one of conjunctions). Furthermore, the size of these layers is linear in terms of the size of $\PP$. 
Similarly, the representation of an unfoundedness refinement will only be quadratically in the size of $\PP$ (quadratically since computing the smallest $y'$ is a refinement takes a linear number of applications). 

The two results now follow from Lemma \ref{lem:nbrefinements:parametrised}, which yields a polynomial upper bound on the number of refinements, and which also allows us to ignore the stop conditions (in general checking whether a fixpoint is reached is a co-NP problem, namely checking equivalence of two circuits; however, we do not need to do this since we have an upper bound on the maximal number of refinements before such a fixpoint is reached).}

In the terminology of \citet{darwiche2002knowledge}, this means that Boolean circuits are \emph{at least as succinct} as logic programs under the parametrised well-founded semantics.
With other languages, for example when \LL denotes OBDDs or SDDs, our algorithm can take exponential time, and its output can take exponential space in the size of \PP.
%
%
%
This is not surprising given the fact these languages support many \text{(co-)NP} hard inference tasks in polynomial time.
Because they support equivalence checking (which is convenient to detect fixpoints early) and have a practically efficient \textsc{Apply} function~\cite{aaai/BroeckD15}, OBDDs and SDDs are excellent languages for use in \compile.

\subsection{Approximate Knowledge Compilation}
The above section provides us with a way to perform various types of inference on logic programs: we can compile any logic program into a target formalism suitable for inference (e.g., SDD for equivalence checking or weighted model counting, CNF for satisfiability checking, etc.). 
However, when working with large programs this approach will be infeasible, simply because compilation is too expensive. 
In this case, we often want to perform approximate knowledge compilation~\cite{selman1996knowledge}.
Well-founded inductions provide us with the means to do this.
\thmwithproof{prop:approx:lp}{Proposition}{proposition}{
Suppose the parametrised well-founded model of \PP is $(\sstruct,\sstruct)$. 
 Let $(\sstruct_{i,1}, \sstruct_{i,2})$ be a well-founded induction of $\ppimcons$. 
 Then for every $i$, 
 $\Th(\sstruct_{i,1})\models\Th(\sstruct)\models \Th(\sstruct_{i,2}).$
}{
 \citet{lpnmr/DeneckerV07} showed that if $(x_i,y_i)_{i\leq \beta}$ is a well-founded induction of $A$ and $(x,y)$ the $A$-well-founded model of $O$, then for every $i\leq \beta$, it holds that 
\[(x_i,y_i)\leqp (x,y).\]
Our proposition immediately follows from this result.}
One application of approximate knowledge compilation is in approximate inference by weighted model counting ($\WMC$)~\cite{DBLP:journals/ai/ChaviraD08} for probabilistic logic programs~\cite{tplp/FierensBRSGTJR15}. Let $\varphi$ be a formula (query) over $\voc$ and $w$ a weight function on $\voc$. Then it follows immediately from Proposition \ref{prop:approx:lp} that 
 \begin{align*}
 \WMC(\Th(\sstruct_{i,1})\land \varphi, w)&\leq \WMC(\PP \land \varphi, w)\leq \WMC(\Th(\sstruct_{i,2})\land \varphi, w).
 \end{align*}
As $\compile(\LL)$ follows a well-founded induction, it can be stopped at any time to obtain an upper and lower bound on the weighted model count (and therefore on the probability of the query). 
In fact, Proposition \ref{prop:approx:lp} can be used to perform \emph{any} (anti)-monotonic inference task approximately.




\section{Conclusion}\label{sec:concl}
In this paper, we presented a novel technique for knowledge compilation of general logic programs; our technique extends previously defined algorithms for positive logic programs.  
Our work is based on the constructive nature of the well-founded semantics:
we showed that the algebraical concept of a well-founded induction translates into a family of anytime knowledge compilation algorithms. 
We used this to show that Boolean circuits are at least as succinct as logic programs (under the parametrised well-founded semantics). 
Our technique also extends to Kripke-Kleene semantics and to other knowledge representation formalisms. 
Extending the implementation by \citet{Jonas} to general logic programs and testing it on a set of benchmarks are topics for future work.

\newpage
\bibliographystyle{acmtrans}

\bibliography{krrlib,customrefs}
 \section{Figures}\label{app:figures}
 This appendix contains some figures associated with the gear wheels example (Example \ref{ex:gear}).  The first figure contains a circuit representation of the parametrised well-founded model of logic program $\PP_w$ from Example \ref{ex:gear}. 

\begin{figure}[htb]
  \centering
  \includegraphics{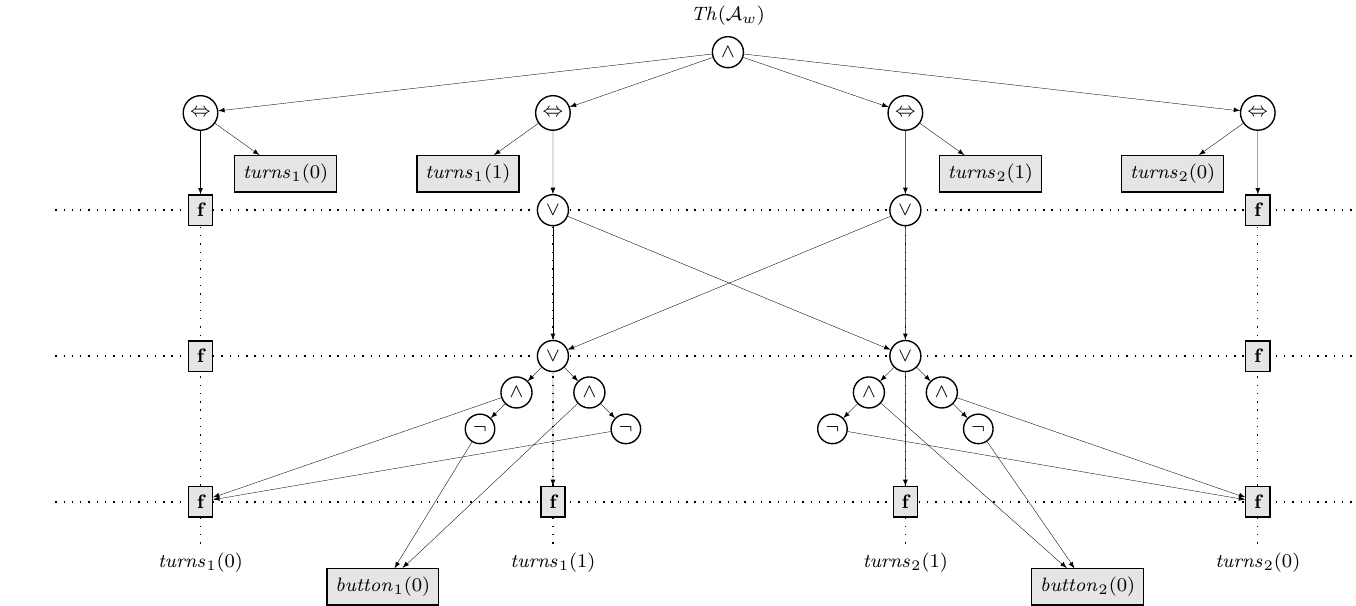}
  \caption{A circuit representation of the gear wheel theory $\Th(\sstruct_w)$.}
  \label{fig:wheel}
\end{figure}

The next figure contains a circuit representation of the parametrised well-founded model of the following logic program $\PP_{w,2}$ that represent the gear wheel example with time ranging from $0$ to $2$:
  \[\left\{
\begin{array}{ll}
\turns_1(0) \lrule \turns_2(0)&
\turns_2(0) \lrule \turns_1(0)\\
\turns_1(1) \lrule \turns_2(1)&
\turns_2(1) \lrule \turns_1(1)\\
\turns_1(2) \lrule \turns_2(2)&
\turns_2(2) \lrule \turns_1(2)\\
\turns_1(1) \lrule \turns_1(0) \land \lnot \button_1(0)&
\turns_2(1) \lrule \turns_2(0) \land \lnot \button_2(0)\\
\turns_1(1) \lrule \lnot \turns_1(0) \land \button_1(0)&
\turns_2(1) \lrule \lnot \turns_2(0) \land \button_2(0)\\
\turns_1(2) \lrule \turns_1(1) \land \lnot \button_1(1)&
\turns_2(2) \lrule \turns_2(1) \land \lnot \button_2(1)\\
\turns_1(2) \lrule \lnot \turns_1(1) \land \button_1(1)&
\turns_2(2) \lrule \lnot \turns_2(1) \land \button_2(1)
 \end{array}
\right\}\]

\clearpage
\begin{figure}[htb]
  \centering
  \includegraphics[angle=-90]{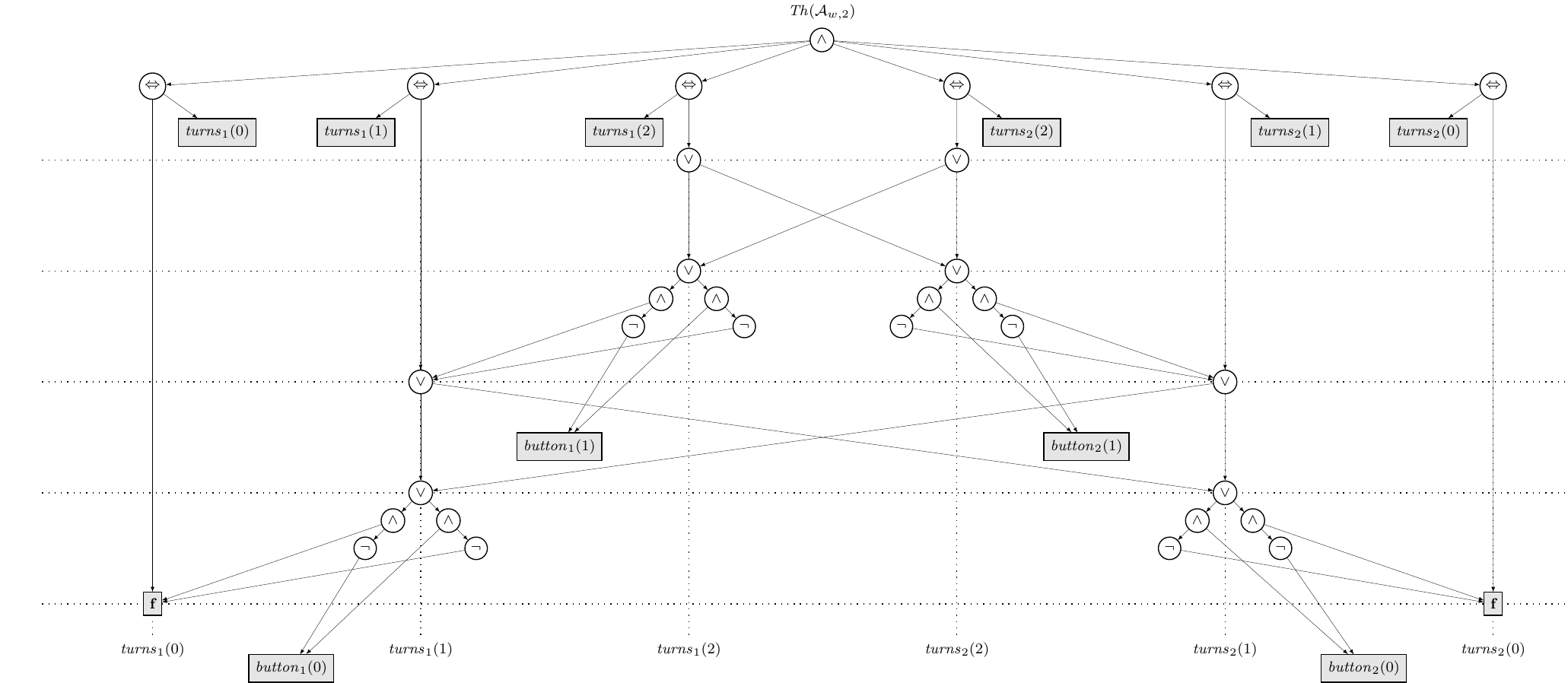}
  \caption{A circuit representation of the gear wheel example for up to two time points.}
  \label{fig:wheel:big}
\end{figure}
\clearpage

 \newpage
 \section{Proofs}
 \renewcommand{\thetheorem}{A.\arabic{theorem}}
 \label{app:proofs}
 \proofs
%
%

\end{document}